\documentclass[journal,twoside,web]{ieeecolor}

\usepackage{generic}
\usepackage{amsmath,amssymb,amsfonts}
\usepackage{graphicx}
\graphicspath{{figures/}}
\usepackage{algorithm,algorithmic}
\usepackage[hidelinks]{hyperref} 
\usepackage{booktabs}
\usepackage{textcomp}
\usepackage{makecell}
\usepackage{multirow}
\usepackage{makecell}
\usepackage{bm}
\usepackage{caption}
\usepackage[style=ieee, citestyle=numeric-comp, sorting=none, sortcites=true]{biblatex}   
\AtBeginBibliography{\footnotesize} 
\AtEveryBibitem{\clearfield{doi}} 
\definecolor{eyefriendlygreen}{RGB}{0,150,80}
\def\BibTeX{{\rm B\kern-.05em{\sc i\kern-.025em b}\kern-.08em
    T\kern-.1667em\lower.7ex\hbox{E}\kern-.125emX}}
\allowdisplaybreaks[4]
\newtheorem{theorem}{Theorem}
\newtheorem{lemma}{Lemma}
\newtheorem{prop}{Proposition}
\newtheorem{rmk}{Remark} 

\begin{document}
\title{DeepONet-LSTM Neural Operator for Output Feedback Control of Reaction Diffusion PDEs}
\author{Jing Zhang, Jie Qi , Linglong Jiang
\thanks{The paper was mainly supported by the National Natural Science Foundation of China (62403305).}
\thanks{A preliminary version of this paper has been accepted for oral presentation at the 23rd IFAC World Congress.}
\thanks{Jing Zhang and Linglong Jiang are with the College of Information Engineering, Shanghai Maritime University, Shanghai  201306, China (e-mail: zhang.jing@shmtu.edu.cn).}
\thanks{Jie Qi is with the School of Information and Intelligence Science, Shanghai 201620, China (e-mail: jieqi@dhu.edu.cn). Jie Qi is the corresponding author.} 
}

\maketitle
 
\begin{abstract}
This paper presents a neural operator-based approach for the output feedback boundary stabilization of reaction diffusion PDEs. The classical output feedback backstepping design requires solving control and observer kernel equations for each reaction coefficient. To avoid computing these kernel functions, the output feedback control law is reformulated as a causal boundary operator that maps the reaction coefficient and the boundary measurement to the boundary control input.  A hybrid DeepONet-LSTM neural operator is proposed to approximate this causal operator, where DeepONet encodes the spatial coefficient and LSTM captures the temporal dependence of the measurement history. We analyze the Lipschitz continuity of the boundary operator and prove the closed-loop  practical stability with the learned controller. A modified loss is also introduced to improve the temporal regularity of the learned boundary input. Numerical results illustrate that the proposed neural operator controller effectively stabilizes the system.
\end{abstract}

\begin{IEEEkeywords}
Reaction diffusion PDE, Backstepping, Neural operator, Output feedback, DeepONet, LSTM.
\end{IEEEkeywords}

\section{Introduction}
\label{sec:introduction}
\IEEEPARstart{B}{ackstepping} output feedback control provides a systematic framework for stabilizing reaction diffusion PDEs from boundary measurements \cite{krstic2008boundary}. For spatially varying reaction coefficients, however, the controller and observer kernels must be recomputed for each coefficient profile. While recent neural-operator approaches approximate backstepping kernels or full-state feedback operators to reduce this offline computation \cite{bhan2024neural,krstic2024neural}, a direct approximation of the observer-based output feedback operator remains largely unexplored. This case is different from full-state feedback, since the control input is generated through the observer dynamics and hence depends causally on the past boundary measurements. 
Motivated by this observation, this paper develops a hybrid DeepONet-LSTM architecture to approximate the coefficient- and history-dependent output feedback operator in a single neural operator, thereby avoiding online kernel recomputation and separate approximation of controller and observer gains.

Deep learning methods have been widely investigated for modeling and prediction of PDE systems; see, e.g., the survey \cite{huang2025partial}. For temporal input-output modeling, LSTM networks provide a recurrent structure for capturing long-term dependencies \cite{hochreiter1997long} and have been used to predict responses of nonlinear dynamical systems \cite{feng2023predicting}. For distributed parameter systems, where the state evolves in both space and time, ConvLSTM-type architectures combine convolutional operations for local spatial correlations with recurrent hidden states for temporal evolution \cite{shi2015convolutional}, and have been applied to sequence-to-sequence prediction of PDE fields \cite{kakka2022sequence}. 

Neural operator methods provide a different framework by learning mappings between function spaces \cite{azizzadenesheli2024neural}. Representative architectures, such as DeepONet \cite{lu2021learning} and Fourier neural operators \cite{li2020fourier}, have been used to approximate PDE solution operators, including coefficient-to-solution and initial-condition-to-solution maps. For time-dependent input signals, however, standard neural operators do not by themselves encode temporal causality for time-dependent input signals. Causality-DeepONet addresses this issue by enforcing that the current response depends only on current and past inputs \cite{liu2024causality}. Another related direction combines neural operators with recurrent architectures; for example, \cite{michalowska2024neural} combines DeepONet and FNO with recurrent modules, including simple RNN, GRU, and LSTM, for long-time prediction of the Korteweg-de Vries dynamics.



Another line of work combines neural operator learning with model-based control design. A pioneering contribution \cite{bhan2024neural} uses neural operators to approximate backstepping kernels and further extends the approximation to full state feedback laws for first order hyperbolic PDE systems. 
A second representative example is \cite{krstic2024neural}, where neural operators are used to approximate both controller and observer gain functions for reaction diffusion PDEs with spatially varying reaction coefficients, which are the same class of systems considered in this paper.

Building on these ideas, neural operator based backstepping designs have been extended in several directions, including adaptive control \cite{bhan2025adaptive} and delay-compensated control  \cite{wang2025deep, zhang2025neural} of reaction diffusion PDEs, adaptive and gain-scheduled control of transport PDEs \cite{lamarque2025adaptive,lamarque2025gain}, and hyperbolic PIDEs with recycle and delay \cite{qi2024neural,qi2025neural}. These works show that neural operators can substantially reduce the computational cost of kernel or gain generation while retaining the closed-loop stability of backstepping designs.

Beyond backstepping, learning-based methods have been explored for PDE control and optimization, including optimal boundary control of parabolic PDEs \cite{sun2025learning}, neural-operator and differentiable predictive control formulations for PDE-constrained optimization \cite{sarkar2025learning}, safe boundary control with barrier functions \cite{hu2025safe}, operator-theoretic optimal control of infinite-dimensional systems \cite{feng2025optimal}, and physics-informed PDE-constrained optimal control \cite{barry2025physics}.

Existing neural-operator-based backstepping designs mainly approximate controller kernels, observer kernels, or instantaneous map from the current state to the current feedback input. In contrast, we directly approximate the observer-based output feedback operator. This operator depends not only on the spatially varying coefficient but also on the measurement history through the observer dynamics. It is therefore a causal history-dependent control operator, rather than a static kernel map or a current-state feedback map.

To capture these two types of dependence, we develop a hybrid DeepONet-LSTM architecture. The DeepONet component encodes the dependence on the spatial coefficient, while the LSTM component recursively encodes the causal measurement history. This enables direct generation of the boundary control input by a single neural operator, without separately approximating the controller and observer kernels or reconstructing the observer online. We prove that the output feedback operator is single-valued and Lipschitz continuous on compact admissible sets, which provides the basis for uniform approximation by the proposed neural operator. A modified training loss is also introduced to improve the temporal regularity of the learned boundary input by penalizing rapid variations and large amplitudes.

The practical stability analysis is carried out by viewing the neural approximation error as a boundary perturbation in the target system. Since this perturbation enters through a nonhomogeneous Dirichlet boundary condition, a direct Lyapunov energy estimate cannot be closed by a standard trace inequality. We instead use a heat kernel boundary convolution estimate for the target heat equation to show that bounded neural approximation errors lead to practical stability of the closed-loop system.

A preliminary version of this work was accepted for presentation at the 2026 IFAC conference \cite{zhang2026hybrid}. Compared with the conference version, this paper provides a complete causal-operator formulation and Lipschitz analysis of the output feedback law, refines the stability proof via a heat-kernel treatment of the boundary approximation error, and introduces a modified loss for improving the temporal regularity of the learned boundary input.

The main contributions of this paper are summarized as follows:
\begin{itemize}
\item A hybrid DeepONet-LSTM neural operator is proposed for observer-based output feedback boundary control of reaction diffusion PDEs. The proposed architecture directly approximates the causal output feedback law with a single neural operator, avoiding online kernel recomputation and separate approximation of the controller and observer gains. A modified training loss is also introduced to improve the temporal regularity and amplitude behavior of the learned boundary input. 

\item The backstepping output feedback law is formulated as a causal history-dependent operator. We prove that this operator is single-valued and Lipschitz continuous on compact admissible sets, which provides the basis for uniform neural operator approximation.

\item Practical stability under neural boundary control is established for bounded approximation errors. The neural approximation error is treated as a nonhomogeneous Dirichlet boundary perturbation in the backstepping target system, and a heat kernel boundary convolution estimate is used to derive the practical stability bound.

\end{itemize}

The remainder of this paper is organized as follows. Section~\ref{sec:backstepping_control} reviews the output feedback backstepping design for the reaction diffusion PDE. Section~\ref{sec:output_feedback_operator} defines the output feedback operator and proves its Lipschitz continuity. Section~\ref{sec:deeponet_lstm} presents the DeepONet-LSTM architecture.
Section~\ref{sec:stab} establishes practical stability under neural boundary control. Section~\ref{sec:simulation} gives numerical results and Section~\ref{sec:conc} concludes this paper.  
\vspace{-3mm}
\subsection*{Notation}
Let $\Omega=\{(x,y):0\le y\le x\le 1\}$ be the triangular kernel domain.
For a continuous function $f$ on $[0,1]$, we use
$\|f\|_\infty=\sup_{x\in[0,1]}|f(x)|$. For a time signal
$p\in C([0,T])$, we write
\begin{equation}
    \|p\|_{C([0,T])}
    :=
    \sup_{t\in[0,T]}|p(t)|.
\end{equation}
For a state $v(\cdot,t)\in L^2(0,1)$, its $L^2$ norm is denoted by
\begin{equation}
    \|v(\cdot,t)\|_{L^2}
    :=
    \left(\int_0^1 |v(x,t)|^2\,dx\right)^{1/2}.
\end{equation}
For a kernel $q$ defined on $\Omega$, we denote
\begin{equation}
    \|q\|_{C(\Omega)}
    :=
    \sup_{(x,y)\in\Omega}|q(x,y)|.
\end{equation}
When the boundary trace $q(1,\cdot)$ is involved, its $L^2(0,1)$ norm is
written as
\begin{equation}
    \|q(1,\cdot)\|_{L^2}
    :=
    \left(\int_0^1 |q(1,y)|^2\,dy\right)^{1/2}.
\end{equation}
All constants denoted by $C$ may vary from line to line, while constants
with subscripts, such as $C_T$ or $C_{\rm ker}$, are used to indicate their
main dependencies.

  
\section{Backstepping Control for a Linear Reaction Diffusion PDE}
\label{sec:backstepping_control}

This paper considers a linear reaction diffusion PDE with a spatially varying
reaction coefficient $\lambda\in C([0,1])$:
\begin{align}
    u_t(x,t)
    &=
    u_{xx}(x,t)+\lambda(x)u(x,t),
    \quad x\in(0,1),\ t>0, \label{eq:main_u}\\
    u(0,t)&=0,   \qquad   u(1,t) =U(t), \label{eq:bc}
\end{align}
where $u_t$ denotes the first-order time derivative and $u_{xx}$ denotes
the second-order spatial derivative. Under full-state feedback, the
classical backstepping boundary control law is given by
\begin{equation}
    U(t)   =
    K_\lambda u(\cdot,t)    :=
    \int_0^1 k[\lambda](1,y)u(y,t)\,dy,
    \label{eq:state_feedback_U}
\end{equation}
where $k[\lambda](x,y)$   is the backstepping control kernel defined on $\Omega$ and satisfies
\begin{align}
    k_{xx}(x,y)-k_{yy}(x,y)
    &=
    \lambda(y)k(x,y), \label{eq:kernel_main}\\
    k(x,0)&=0, \label{eq:kernel_bc1}\\
    k(x,x)
    &=
    -\frac{1}{2}\int_0^x \lambda(y)\,dy.
    \label{eq:kernel_bc2}
\end{align}
This is the standard boundary backstepping design for reaction diffusion
PDEs \cite{krstic2008boundary,smyshlyaev2004closed}.

To implement output feedback, we assume that the boundary flux
\begin{equation}
    p(t)=u_x(0,t)
    \label{eq:measurement_p}
\end{equation}
is available for measurement. The observer is designed according to the
backstepping observer framework \cite{smyshlyaev2005backstepping}:
\begin{align}
    \hat u_t(x,t)
    &=
    \hat u_{xx}(x,t)+\lambda(x)\hat u(x,t) \nonumber\\
    & \qquad \qquad\quad +g_1[\lambda](x)\big(p(t)-\hat u_x(0,t)\big),
    \label{eq:observer}\\
    \hat u(0,t)&=0, \qquad    
    \hat u(1,t) =U(t). \label{eq:observer_bc_right}
\end{align}
The observer gain is given by $
    g_1[\lambda](x)=-g[\lambda](x,0)$, 
where $g[\lambda](x,y)$ is the observer kernel. It is also defined on
$\Omega$ and satisfies
\begin{align}
    g_{xx}(x,y)-g_{yy}(x,y)
    &=
    -\lambda(x)g(x,y), \label{eq:observer_kernel_main}\\
    g(1,y)&=0, \label{eq:observer_kernel_bc1}\\
    g(x,x)
    &=
    -\frac{1}{2}\int_x^1 \lambda(y)\,dy.
    \label{eq:observer_kernel_bc2}
\end{align}
The resulting observer-based output feedback control law is
\begin{equation}
    U(t)
    =
    K_\lambda \hat u(\cdot,t)
    :=
    \int_0^1 k[\lambda](1,y)\hat u(y,t)\,dy.
    \label{eq:output_feedback_U}
\end{equation}

\section{Output Feedback Operator}
\label{sec:output_feedback_operator}
Existing deep operator learning approaches for backstepping control have mainly focused on approximating the control and observer kernels \cite{krstic2024neural}, or the full-state feedback control operator \cite{bhan2024neural}. In contrast, the output feedback operator induced by a backstepping observer has not been directly addressed.
Unlike full-state feedback, the output feedback law is not an instantaneous mapping from the current measurement $p(t)$ to the current boundary input $U(t)$. Since $U(t)$ in \eqref{eq:output_feedback_U} is computed from the observer state $\hat u(\cdot,t)$, and $\hat u(\cdot,t)$ evolves according to the dynamic observer \eqref{eq:observer}-\eqref{eq:observer_bc_right}, the control input depends on the measurement history rather than only on the current value $p(t)$. In particular, the observer state is generated by dynamically filtering the  signal $p(s)-\hat u_x(0,s)$ over $0\le s\le t$. 

To make this dependence explicit, we next describe the observer in an abstract evolution form.
\begin{equation}
\frac{d}{dt}\hat u(t)
=
\mathcal A_\lambda \hat u(t)
+
\mathcal G_\lambda p(t),
\label{eq:abstract_observer}
\end{equation}
where $\mathcal A_\lambda$ denotes the observer generator, including the boundary feedback and output injection terms, and $\mathcal G_\lambda$ represents the measurement injection operator. Then the mild solution can be formally expressed as
\begin{equation}
\hat u(t)
=
e^{\mathcal A_\lambda t}\hat u_0
+
\int_0^t
e^{\mathcal A_\lambda(t-s)}
\mathcal G_\lambda p(s) ds.
\label{eq:observer_mild_solution}
\end{equation}
This representation shows explicitly that $\hat u(\cdot,t)$ depends on the entire past measurement trajectory $p_{[0,t]}
    :=
    \{p(s):0\le s\le t\}$, rather than only on the current value $p(t)$. Then, for each reaction coefficient $\lambda(x)$, the controller induces a causal output feedback operator
\begin{equation}
U(t)
=
\mathcal T[\lambda,p_{[0,t]}](t),
\end{equation}
where 
\begin{equation}
    \mathcal T[\lambda,p_{[0,t]}](t)
    :=
    \int_0^1 k[\lambda](1,y)\hat u_{\lambda,p}(y,t)\,dy.
    \label{eq:full_history_operator_def}
\end{equation}
\subsection{Single-valuedness of the output feedback operator}
\begin{prop}[Single-valuedness of the operator]
\label{prop:full_history_single_valued}
Let $T>0$. Assume that $\lambda\in C([0,1])$, $p\in C([0,T])$, and the
observer initial condition $\hat u_0\in L^2(0,1)$ is prescribed and fixed.
Then, for every admissible pair $(\lambda,p_{[0,t]})$, the boundary input
$U(t)$ is uniquely determined, that is,
$\mathcal T$ defined in \eqref{eq:full_history_operator_def} is a
single-valued operator.
    \end{prop}
\begin{proof}
For a fixed $\lambda$, the backstepping kernel equations uniquely determine
the control kernel $k[\lambda]$ and the observer gain $g_1[\lambda]$.
By the well-posedness of the
backstepping observer system in $L^2(0,1)$
\cite{smyshlyaev2005backstepping}, it admits a unique
solution $\hat u_{\lambda,p}(\cdot,t)$ on $[0,T]$ for given $p_{[0,t]}$ and a fixed initial
condition $\hat u_0$ . Therefore, $U(t)$ defined by \eqref{eq:output_feedback_U} 
is uniquely determined for each $t\in[0,T]$. This proves the proposition.   
\end{proof}

\begin{rmk}
The plant initial condition $u(\cdot,0)$ is not explicitly included in
the operator input. Its effect is reflected in the measured output history
$p_{[0,t]}$. Therefore, for the observer-induced output feedback law, the
mapping $(\lambda,p_{[0,t]})\mapsto U(t)$ can cover different plant
initial conditions.
\end{rmk}

\subsection{Continuity of the output feedback operator}
\label{subsec:continuity_output_feedback_operator}

Let
\begin{equation}
    \mathcal K_\lambda
    :=
    \{\lambda\in C([0,1]):\|\lambda\|_\infty\le B_\lambda\},
\end{equation}
and
\begin{equation}
    \mathcal K_p
    :=
    \{p\in C([0,T]):\|p\|_{C([0,T])}\le B_p\},
\end{equation}
where $B_\lambda>0$ and $B_p>0$ are constants.

\begin{lemma}[Lipschitz dependence of the kernels on $\lambda$]
\label{lem:kernel_lipschitz}
For every $B_\lambda>0$, there exists a constant
$C_{\rm ker}=C_{\rm ker}(B_\lambda)>0$ such that, for all
$\lambda_1,\lambda_2\in\mathcal K_\lambda$,
\begin{align}
    \|k[\lambda_1]-k[\lambda_2]\|_{C(\Omega)}
    &\le
    C_{\rm ker}\|\lambda_1-\lambda_2\|_\infty,
    \label{eq:k_lipschitz}\\
    \|g[\lambda_1]-g[\lambda_2]\|_{C(\Omega)}
    &\le
    C_{\rm ker}\|\lambda_1-\lambda_2\|_\infty.
    \label{eq:g_lipschitz}
\end{align}
Consequently, since $g_1[\lambda](x)=-g[\lambda](x,0)$,
\begin{equation}
    \|g_1[\lambda_1]-g_1[\lambda_2]\|_{\infty}
    \le
    C_{\rm ker}\|\lambda_1-\lambda_2\|_\infty .
    \label{eq:g1_lipschitz}
\end{equation}
Moreover, there exists $C_{\rm ker}^0=C_{\rm ker}^0(B_\lambda)>0$ such that
\begin{equation}
    \|k[\lambda](1,\cdot)\|_{L^2}
    +
    \|g_1[\lambda]\|_{\infty}
    \le C_{\rm ker}^0,
    \qquad \forall \lambda\in\mathcal K_\lambda .
    \label{eq:kernel_uniform_bound}
\end{equation}
\end{lemma}

\begin{proof}
The control kernel $k[\lambda]$ and the observer kernel $g[\lambda]$ satisfy linear hyperbolic kernel equations on the triangular domain $\Omega$. Their solutions admit equivalent integral representations and are known to be uniformly bounded \cite{smyshlyaev2004closed, smyshlyaev2005backstepping}. 
Applying the Volterra-Gronwall estimate to the difference of the two integral equations associated with $\lambda_1$ and $\lambda_2$, one obtains \eqref{eq:k_lipschitz}-\eqref{eq:g_lipschitz}. The estimate \eqref{eq:g1_lipschitz} follows directly from $g_1[\lambda](x)=-g[\lambda](x,0)$. The uniform boundedness \eqref{eq:kernel_uniform_bound} follows by taking $\lambda_2=0$ and using the boundedness of $\mathcal K_\lambda$. 
\end{proof}


\begin{lemma}[Lipschitz dependence of the observer state]
\label{lem:observer_lipschitz}
Let $T>0$. For any $B_\lambda>0$, $B_p>0$, and fixed
$\hat u_0\in L^2(0,1)$, there exists a constant
$C_{\rm obs}=C_{\rm obs}(T,B_\lambda,B_p,\|\hat u_0\|_{L^2})>0$ such that,
for any $
    \lambda_1,\lambda_2\in\mathcal K_\lambda$, and
    $p_1,p_2\in\mathcal K_p$, 
the corresponding observer states  $\hat u_i:=\hat u_{\lambda_i,p_i}$, $i=1,2$,
satisfy
\begin{align}
    \|\hat u_1-\hat u_2\|&_{C([0,T];L^2)}
    \le
    \nonumber\\
    &C_{\rm obs}
    \left(
    \|\lambda_1-\lambda_2\|_\infty
    +
    \|p_1-p_2\|_{C([0,T])}
    \right).
    \label{eq:observer_lipschitz_final}
\end{align}
In addition, there exists $C_{\rm obs}^0>0$ such that
\begin{equation}
    \|\hat u_i\|_{C([0,T];L^2)}
    \le
    C_{\rm obs}^0,
    \qquad i=1,2.
    \label{eq:observer_uniform_bound}
\end{equation}
\end{lemma}

\begin{proof}
Let $e(x,t):=\hat u_1(x,t)-\hat u_2(x,t)$. Subtracting the two observer systems gives
\begin{align}
    e_t
    &=
    e_{xx}+\lambda_1(x)e
    -g_1[\lambda_1](x)e_x(0,t)
    \notag\\&\quad
    +g_1[\lambda_1](x)(p_1(t)-p_2(t))
    \notag\\
    &\quad
    +
    \big(g_1[\lambda_1](x)-g_1[\lambda_2](x)\big)
    \big(p_2(t)-\hat u_{2x}(0,t)\big)
     \notag\\&\quad
    +
    \big(\lambda_1(x)-\lambda_2(x)\big)\hat u_2(x,t),
    \label{eq:e_observer_difference}
\end{align}
with initial and boundary conditions $e(x,0)=0$, $e(0,t)=0$, and
\begin{align}
e(1,t)
&=
\int_0^1 k[\lambda_1](1,y)e(y,t) dy
\notag\\\
&+
\int_0^1
\big(k[\lambda_1](1,y)-k[\lambda_2](1,y)\big)
\hat u_2(y,t) dy .
\label{eq:e_boundary}
\end{align}

The homogeneous part of \eqref{eq:e_observer_difference}, including the boundary feedback \eqref{eq:e_boundary} involving $k[\lambda_1]$, defines a linear parabolic closed-loop generator on $L^2(0,1)$. The remaining terms are distributed inputs 
\begin{equation}
    g_1[\lambda_1](p_1-p_2),
   ~
    (g_1[\lambda_1]-g_1[\lambda_2])(p_2-\hat u_{2x}(0,t)),
   ~
    (\lambda_1-\lambda_2)\hat u_2, \nonumber
\end{equation} 
and a boundary perturbation in \eqref{eq:e_boundary}. Hence, by the continuous-dependence estimate for linear parabolic systems with admissible boundary inputs \cite{curtain1995infinite}, the mild solution satisfies
\begin{align}
    \|e\|_{C([0,T];L^2)}
    &\le
    C_T\Big(
    \|p_1-p_2\|_{C([0,T])}
    +
    \|g_1[\lambda_1]-g_1[\lambda_2]\|_{\infty}
    \notag\\
    +~&
    \|\lambda_1-\lambda_2\|_\infty
    +
    \|k[\lambda_1]-k[\lambda_2]\|_{C(\Omega)}
    \Big),
    \label{eq:e_energy_estimate}
\end{align}
where $C_T>0$ depends on $T$, $B_\lambda$, $B_p$, and
$\|\hat u_0\|_{L^2}$, but not on the particular choices of
$\lambda_i$ and $p_i$.

Using Lemma~\ref{lem:kernel_lipschitz}, we have 
\begin{equation}
    \|e\|_{C([0,T];L^2)}
    \le
    C_{\rm obs}
    \left(
    \|\lambda_1-\lambda_2\|_\infty
    +
    \|p_1-p_2\|_{C([0,T])}
    \right), \nonumber
\end{equation}
which proves \eqref{eq:observer_lipschitz_final}.

The uniform bound \eqref{eq:observer_uniform_bound} follows similarly from
the well-posedness estimate for the observer system, the boundedness of
$\lambda$, $p$, $k[\lambda]$, and $g_1[\lambda]$, and the fixed
initial condition $\hat u_0$.
\end{proof}

\begin{prop}[Lipschitz continuity of the output operator]
\label{prop:T_lipschitz}
Let $T>0$, $B_\lambda>0$, and $B_p>0$. Under the assumptions of
Lemmas~\ref{lem:kernel_lipschitz} and \ref{lem:observer_lipschitz}, the
output feedback operator \eqref{eq:full_history_operator_def} is locally Lipschitz continuous, i.e., there exists a constant
\[
    C_{\mathcal T}
    =
    C_{\mathcal T}(T,B_\lambda,B_p,\|\hat u_0\|_{L^2})>0
\]
such that, for any 
    $\lambda_1,\lambda_2\in\mathcal K_\lambda$, and
    $p_1,p_2\in\mathcal K_p$, 
one has
\begin{align}
    &\|\mathcal T[\lambda_1,p_1]
    -
    \mathcal T[\lambda_2,p_2]\|_{C([0,T])}
   \nonumber 
   \\\le&
    C_{\mathcal T}
    \left(
    \|\lambda_1-\lambda_2\|_\infty
    +
    \|p_1-p_2\|_{C([0,T])}
    \right).
    \label{eq:T_lipschitz_estimate}
\end{align}
\end{prop}

\begin{proof}
Define    $ U_i(t):=\mathcal T[\lambda_i,p_i](t)$. 
Then
\begin{align}
    U_1(t)-&U_2(t)
    =
    \int_0^1
    \big(k[\lambda_1](1,y)-k[\lambda_2](1,y)\big)
    \hat u_1(y,t)\,dy
    \notag\\
    &\quad+
    \int_0^1
    k[\lambda_2](1,y)
    \big(\hat u_1(y,t)-\hat u_2(y,t)\big)\,dy .
    \label{eq:U_difference_split}
\end{align}
Using the Cauchy-Schwarz inequality gives
\begin{align}
    |U_1(t)-&U_2(t)|
    \le
    \|k[\lambda_1](1,\cdot)-k[\lambda_2](1,\cdot)\|_{L^2}
    \|\hat u_1(\cdot,t)\|_{L^2}
    \notag\\
    &\quad+
    \|k[\lambda_2](1,\cdot)\|_{L^2}
    \|\hat u_1(\cdot,t)-\hat u_2(\cdot,t)\|_{L^2}.
    \label{eq:U_difference_bound}
\end{align}
Taking the supremum over $t\in[0,T]$, and using
Lemmas~\ref{lem:kernel_lipschitz} and \ref{lem:observer_lipschitz}, we obtain \eqref{eq:T_lipschitz_estimate} with
$ C_{\mathcal T}=C_{\rm ker}C^0_{\rm obs}+C^0_{\rm ker}C_{\rm obs}$.
\end{proof}

\section{DeepONet-LSTM Neural Operator Architecture}
\label{sec:deeponet_lstm}

\subsection{DeepONet-LSTM operator}

To approximate the output feedback causal operator
\eqref{eq:full_history_operator_def}, we introduce a hybrid DeepONet-LSTM
neural operator. 
A standard DeepONet is not directly suitable for approximating this causal operator, since it does
not explicitly encode the temporal dependence of the measurement history.
Therefore, we combine a
DeepONet-based encoder for the spatial coefficient with an LSTM module for
the temporal measurement history, as shown in Fig.~\ref{fig:net}.
\begin{figure*}[!htp]
\centering
\includegraphics[width=0.8\textwidth]{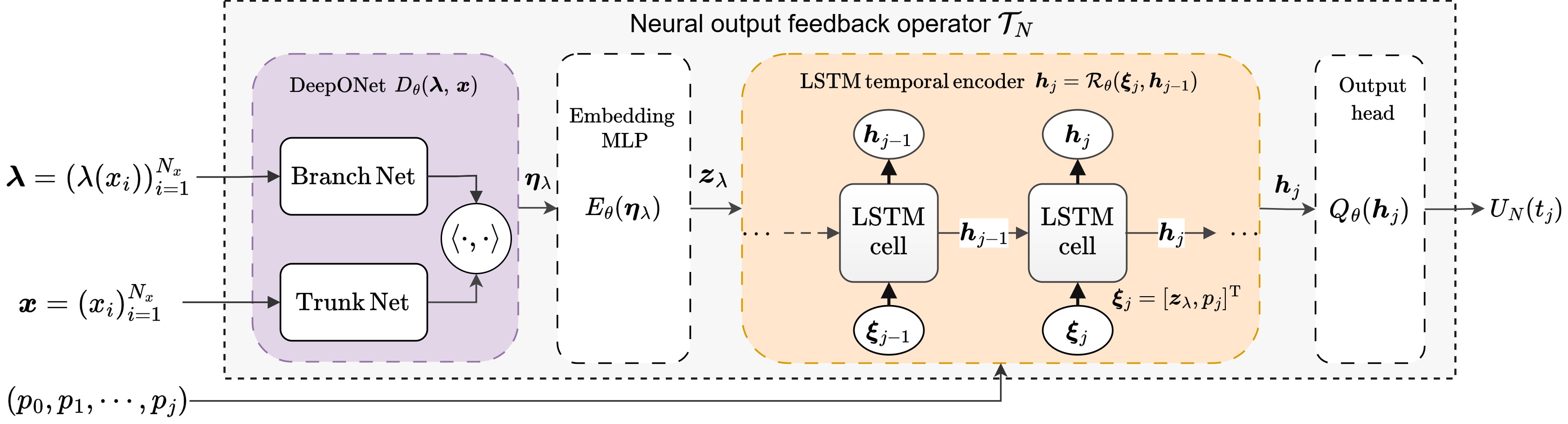}
\caption{Schematic of the proposed   DeepONet-LSTM architecture for the  output feedback operator.}
\label{fig:net}
\end{figure*}

The DeepONet component encodes the spatially varying reaction coefficient
$\lambda(x)$. Let $\{x_i\}_{i=1}^{N_x}\subset[0,1]$ be the spatial sensor
points and define
$
\boldsymbol\lambda
:=
\big(\lambda(x_1),\lambda(x_2),\ldots,\lambda(x_{N_x})\big).
$
The branch network takes $\boldsymbol\lambda$ as input, while the trunk
network takes the spatial coordinates.  Their outputs are combined in the standard DeepONet form to obtain 
\begin{equation}
    \boldsymbol \eta_\lambda = D_\theta(\boldsymbol\lambda) \in \mathbb R^{N_x}, 
\end{equation}where $D_\theta$ denotes the DeepONet encoder for the coefficient function.    The resulting DeepONet representation is then mapped by an embedding network to a finite dimensional coefficient embedding \begin{equation} \boldsymbol z_\lambda = E_\theta(\boldsymbol \eta_\lambda) \in \mathbb R^{d_\lambda}, \label{eq:lambda_embedding} \end{equation}
where $E_\theta$ denotes the embedding network.
Since $\lambda(x)$ is time invariant, $z_\lambda$ is also time invariant and
serves as a conditioning input for the temporal module.

The LSTM component is used to encode the causal dependence on the measured
boundary flux. In the discrete-time implementation, let
$t_j=j\Delta t$ and $p_j=p(t_j)$. At each time step, the coefficient embedding
$z_\lambda$ is concatenated with the current measurement $p_j$ to form
\begin{equation}
\boldsymbol \xi_j=[\boldsymbol z_\lambda,p_j]^{\mathrm{T}} \in \mathbb R^{1+d_\lambda},
\label{eq:lstm_input}
\end{equation}
and the LSTM hidden state is updated by
\begin{equation}
	\boldsymbol h_j
	=
	\mathcal R_\theta(\boldsymbol \xi_j,\boldsymbol h_{j-1}),
	\qquad j=0,1,\ldots,N_T,
	\label{eq:lstm_update}
\end{equation}
where the initial state is set to zero.
The neural boundary control input is then generated from the
hidden state by an output network:
\begin{equation}
U_N(t_j)
=
Q_\theta(\boldsymbol h_j) .
\label{eq:lstm_output_control}
\end{equation}
Combining the coefficient encoder, the LSTM recursion, and the output
network yields the discrete-time neural operator representation
\begin{equation}
U_N(t_j)
=
\mathcal T_N[\lambda,p_{[0,t_j]}](t_j),
\label{eq:deeponet_lstm_discrete_operator}
\end{equation}
where  $
\mathcal T_N:
(\lambda,p_{[0,t]})\mapsto U_N(t)$ denotes the DeepONet-LSTM neural operator.
Thus, although only the current measurement $p_j$ is explicitly fed into the
LSTM at time $t_j$, the recursively updated hidden state $h_j$ encodes the history $p_{[0,t_j]}$. 

\begin{theorem}[Approximation of the causal operator]
\label{thm:of_causal_operator_approx}
Let $\mathcal T$ be the output feedback causal operator defined in
\eqref{eq:full_history_operator_def}. For any $\varepsilon>0$, there exists
a DeepONet-LSTM neural operator $\mathcal T_N$ defined in
\eqref{eq:deeponet_lstm_discrete_operator} such that
\begin{equation}
\sup_{(\lambda,p)\in\mathcal K_\lambda\times\mathcal K_p}
|\mathcal T[\lambda,p]-\mathcal T_N[\lambda,p]|_{C([0,T])}
<\varepsilon .
\label{eq:of_operator_approx_error}
\end{equation}
Consequently, for every admissible pair $(\lambda,p)$ and every
$t\in[0,T]$,
\begin{equation}
\left|
\mathcal T[\lambda,p](t)
-
\mathcal T_N[\lambda,p](t)
\right|
<\varepsilon .
\label{eq:of_operator_pointwise_error}
\end{equation}
\end{theorem}

\begin{proof}
By Proposition~\ref{prop:T_lipschitz}, the operator $\mathcal T$ is a
single-valued Lipschitz continuous operator on
$\mathcal K_\lambda\times\mathcal K_p$, and hence is continuous on this
admissible input set. The DeepONet component approximates the dependence on
the spatial function $\lambda(x)$, while the LSTM component approximates the
causal dependence on the measurement history $p_{[0,t]}$. By the universal
approximation property of DeepONet for continuous operators
\cite{chen1995universal,lu2021learning} and   recurrent neural networks for continuous causal maps
\cite{schafer2006recurrent}, the composed DeepONet-LSTM architecture can
approximate $\mathcal T$ uniformly on
$\mathcal K_\lambda\times\mathcal K_p$. Therefore, for any
$\varepsilon>0$, the network size can be chosen sufficiently large such that
\eqref{eq:of_operator_approx_error} holds. The pointwise estimate
\eqref{eq:of_operator_pointwise_error} follows directly from the
$C([0,T])$-norm bound.
\end{proof}


\subsection{Loss Function}
Before introducing the loss function, we first describe the training dataset. The dataset is denoted by
\begin{equation}
 \mathcal D=
\left\{
\left(\lambda^{(i)},p^{(i)},U^{(i)}\right)
\right\}_{i=1}^{N_s},   
\end{equation}
where $N_s$ is the number of training samples, $\lambda^{(i)}$ is the reaction coefficient, and $U^{(i)}(t_k)$ is the reference boundary control input defined by \eqref{eq:output_feedback_U}. The measured boundary signal $p^{(i)}$ is generated from the corresponding closed-loop response $u^{(i)}(x,t)$ under the controller \eqref{eq:output_feedback_U}, namely
\begin{equation}
p^{(i)}(t_k)=u_x^{(i)}(0,t_k),
\qquad k=0,\ldots,M-1 .
\end{equation}
In the implementation, the sampled signals are normalized using the statistics of the training dataset.


To reduce rapid temporal variations and large amplitudes in the learned boundary input, we introduce the modified loss
\begin{equation}
\mathcal L_{\mathrm{mod}}=
\mathcal L_U
+
\alpha \mathcal L_{\Delta U}
+
\beta \mathcal L_{\mathrm{amp}},
\label{eq:loss_modified}
\end{equation}
where $\alpha\geq0$ and $\beta\geq0$ are weighting coefficients. The first term is the standard mean squared error loss,
\begin{equation}
\mathcal L_U=
\frac{1}{N_s M}
\sum_{i=1}^{N_s}
\sum_{k=0}^{M-1}
\left|
U_N^{(i)}(t_k)-U^{(i)}(t_k)
\right|^2 .
\label{eq:loss_mse}
\end{equation}
The second term penalizes rapid temporal variations between adjacent sampling instants:
\begin{equation}
\mathcal L_{\Delta U}\!=\!
\frac{1}{N_s(M-1)}
\sum_{i=1}^{N_s}
\sum_{k=1}^{M-1}
\left|
U_N^{(i)}(t_k) - U_N^{(i)}(t_{k-1})
\right|^2 .
\label{eq:loss_deltaU}
\end{equation}
The third term penalizes large control amplitudes:
\begin{equation}
\mathcal L_{\mathrm{amp}}=\frac{1}{N_s M}
\sum_{i=1}^{N_s}
\sum_{k=0}^{M-1}
\left|
U_N^{(i)}(t_k)
\right|^2 .
\label{eq:loss_amp}
\end{equation}

\section{Practical Stability under Neural Boundary Control}
\label{sec:stab}

This section analyzes the closed-loop stability when the exact output
feedback causal operator $\mathcal T[\lambda,p]$ is replaced by its neural
approximation $\mathcal T_N[\lambda,p]$. Define the induced boundary
approximation error by
\begin{equation}
    \Delta U(t)
    :=
    \mathcal T_N[\lambda,p](t)-\mathcal T[\lambda,p](t).
    \label{eq:Delta_U}
\end{equation}
Under the exact output feedback controller, the backstepping transformation
maps the closed-loop system into a target heat equation
with homogeneous boundary conditions.  When the neural boundary operator is used instead, the
only mismatch is introduced at the controlled boundary. Consequently, the
transformed target system is given by
\begin{align}
    w_t(x,t)&=w_{xx}(x,t),
    \qquad x\in(0,1),\ t>0, \label{eq:target_main}\\
    w(0,t)&=0, 
    \qquad w(1,t) =\Delta U(t). \label{eq:target_bc1}
\end{align}
When $\Delta U(t)\equiv 0$, this system coincides with the nominal target
system and is exponentially stable.

The following lemma gives a fading-memory estimate for the perturbed target
system. 
\begin{lemma}[Estimate for the target system]
\label{lem:stab_w}
Consider the target system \eqref{eq:target_main}-\eqref{eq:target_bc1}. Assume that $|\Delta U(t)|\le \varepsilon$.
Then there exist constants $C_1,~C_2>0$ and $\alpha_0>0$ such that, for all
$t\ge 0$,
\begin{equation}
\|w(\cdot,t)\|_{L^2}^2
\le
C_1 e^{-\alpha_0 t}\|w(\cdot,0)\|_{L^2}^2
+
C_2 \varepsilon^2 .
\label{eq:w_eps_bound}
\end{equation}
\end{lemma}

\begin{proof}
Denote by $G_D(x,y,t)$ the Dirichlet heat kernel on $(0,1)$:
\begin{equation}
G_D(x,y,t)
=
2\sum_{n=1}^{\infty}
e^{-n^2\pi^2 t}
\sin(n\pi x)\sin(n\pi y).
\end{equation}
The solution of \eqref{eq:target_main}-\eqref{eq:target_bc1} can be
represented as
\begin{equation} \nonumber
w 
= 
\int_0^1 \!G_D(x,y,t)w(y,0)dy 
  -\!
\int_0^t\!\!
\partial_y  G_D(x,1,t-\tau) \Delta U(\tau)d\tau .
\end{equation}
The first term is the homogeneous heat semigroup response. Since the first
eigenvalue of the Dirichlet Laplacian on $(0,1)$ is $\pi^2$, we have
\begin{equation}
\left\|
\int_0^1 G_D(\cdot,y,t)w(y,0)dy
\right\|_{L^2}^2
\le
e^{-2\pi^2 t}\|w(\cdot,0)\|_{L^2}^2 .
\label{eq:heat_kernel_initial_estimate}
\end{equation}

Let $P(x,t) =
 -\partial_y G_D(x,1,t)$. Then we obtain
\begin{equation}
P(x,t)
=
2\sum_{n=1}^{\infty}
n\pi (-1)^{n+1}
e^{-n^2\pi^2 t}
\sin(n\pi x).
\end{equation}
Using   Parseval's theorem, one has
\begin{equation}
\|P(\cdot,t)\|_{L^2}^2
=
2\pi^2
\sum_{n=1}^{\infty}
n^2 e^{-2n^2\pi^2 t}.
\end{equation}
Hence there exists a constant $C_P>0$ such that, for all $t>0$,
\begin{equation}
\|P(\cdot,t)\|_{L^2}
\le
C_P t^{-3/4}e^{-\pi^2 t/2}.
\label{eq:boundary_kernel_estimate}
\end{equation}
Therefore, by Minkowski's inequality and $| \Delta U(t)|\le\varepsilon$,
\begin{align}
&\left\|  
\int_0^t \!\! P(\cdot,t-\tau) \Delta U(\tau)d\tau
\right\|_{L^2}
\!\le  \!
\int_0^t \!\!
\|P(\cdot,t-\tau)\|_{L^2}| \Delta U(\tau)|d\tau
\nonumber\\
&\qquad \le
C_P\varepsilon
\int_0^t
(t-\tau)^{-3/4}
e^{-\pi^2(t-\tau)/2}d\tau
  \le
C_b\varepsilon,\label{eq:boundary_convolution_bound}
\end{align}
where $
    C_b   =
    C_P
    \int_0^\infty r^{-3/4}e^{-\pi^2 r/2}\,dr
    <\infty .
$


Combining \eqref{eq:heat_kernel_initial_estimate} and
\eqref{eq:boundary_convolution_bound} gives
\begin{align}
\|w(\cdot,t)\|_{L^2}^2
&\le
2e^{-2\pi^2 t}\|w(\cdot,0)\|_{L^2}^2
+
2C_b^2\varepsilon^2.
\end{align}
Thus, \eqref{eq:w_eps_bound} holds with $C_1=2$, $\alpha_0=2\pi^2$ and $C_2=2C_b^2$. 
\end{proof}

\begin{theorem}[Practical stability under neural control]
\label{thm:practical_stab}
Under the neural boundary control $U_N=\mathcal T_N[\lambda,p_{[0,t]}](t)$, the closed-loop system is
practically stable, that is, there exist constants
$M_1,\alpha_1,\kappa_1>0$ such that, for all $t\ge 0$,
\begin{equation}
\|u(\cdot,t)\|_{L^2}^2
\le
M_1e^{-\alpha_1 t}\|u(\cdot,0)\|_{L^2}^2
+
\kappa_1\varepsilon^2 .
\label{eq:stsb_practical_fading}
\end{equation}
Consequently,
\begin{equation}
\limsup_{t\to\infty}\|u(\cdot,t)\|_{L^2}^2
\le
\kappa_1\varepsilon^2 .
\label{eq:ultimate_bound_corrected}
\end{equation}
\end{theorem}

\begin{proof}
Since the Volterra backstepping transformation and its inverse are bounded
on $L^2(0,1)$, there exist constants $m_1,m_2>0$ such that
\begin{equation}
m_1\|u(\cdot,t)\|_{L^2}^2
\le
\|w(\cdot,t)\|_{L^2}^2
\le
m_2\|u(\cdot,t)\|_{L^2}^2 .
\label{eq:norm_equivalence}
\end{equation}
Combining \eqref{eq:norm_equivalence} and \eqref{eq:w_eps_bound} in Lemma~\ref{lem:stab_w} gives
\begin{align}
\|u(\cdot,t)\|_{L^2}^2
&\le
\frac{C_1m_2}{m_1}e^{-\alpha_0 t}
\|u(\cdot,0)\|_{L^2}^2
+
\frac{C_2}{m_1}\varepsilon^2 .
\end{align}
Therefore, \eqref{eq:stsb_practical_fading} holds by setting
\begin{equation}
M_1=\frac{C_1 m_2}{m_1},
\qquad
\alpha_1=\alpha_0,
\qquad
\kappa_1=\frac{C_2}{m_1}.
\end{equation}
Taking the limit superior as $t\to\infty$ yields
\eqref{eq:ultimate_bound_corrected}.
\end{proof}

\section{Simulation}
\label{sec:simulation}
\subsection{Training Settings}
The training dataset is generated by simulating the nominal output feedback backstepping closed-loop system under different reaction coefficients and initial conditions. We sample $1000$ reaction coefficients of the form $\lambda(x) = 50 \cos(\nu  \arccos(x))$, where $\nu$ follows the uniform distribution on $[4,\,9]$. For each reaction coefficient, we generate $10$ smooth initial conditions satisfying $u(0,0)=u(1,0)=0$, and record a $T=2\,\mathrm{s}$ closed-loop response for each pair of reaction coefficient and initial condition. This gives $N_s=10000$ closed-loop response trajectories in total, forming the dataset $\mathcal D$.  
All variables are normalized using the statistics computed from the training set. During training, the complete sequence is fed into the LSTM. The same network architecture is trained with either the MSE loss
or the modified loss, and both trained models are evaluated using the same
data split. In the modified loss, the weights are chosen as $\alpha=5\times 10^{-3}$ and $\beta=10^{-7}$. The proposed neural operator contains approximately $2.38\times 10^{5}$ trainable parameters.  The two models are trained separately on a single RTX 5090 GPU for up to $1000$ epochs, and each training run takes approximately $10$ min. Fig.~\ref{fig:loss} presents the training and validation loss histories under the modified loss, showing that the proposed DeepONet-LSTM architecture can be trained stably. The training code is available on 
\href{https://github.com/JingZhang-JZ/deeponet-lstm-rd-output-feedback}{GitHub}.
\begin{figure}[!htp]
    \centering
\includegraphics[width=0.75\linewidth]{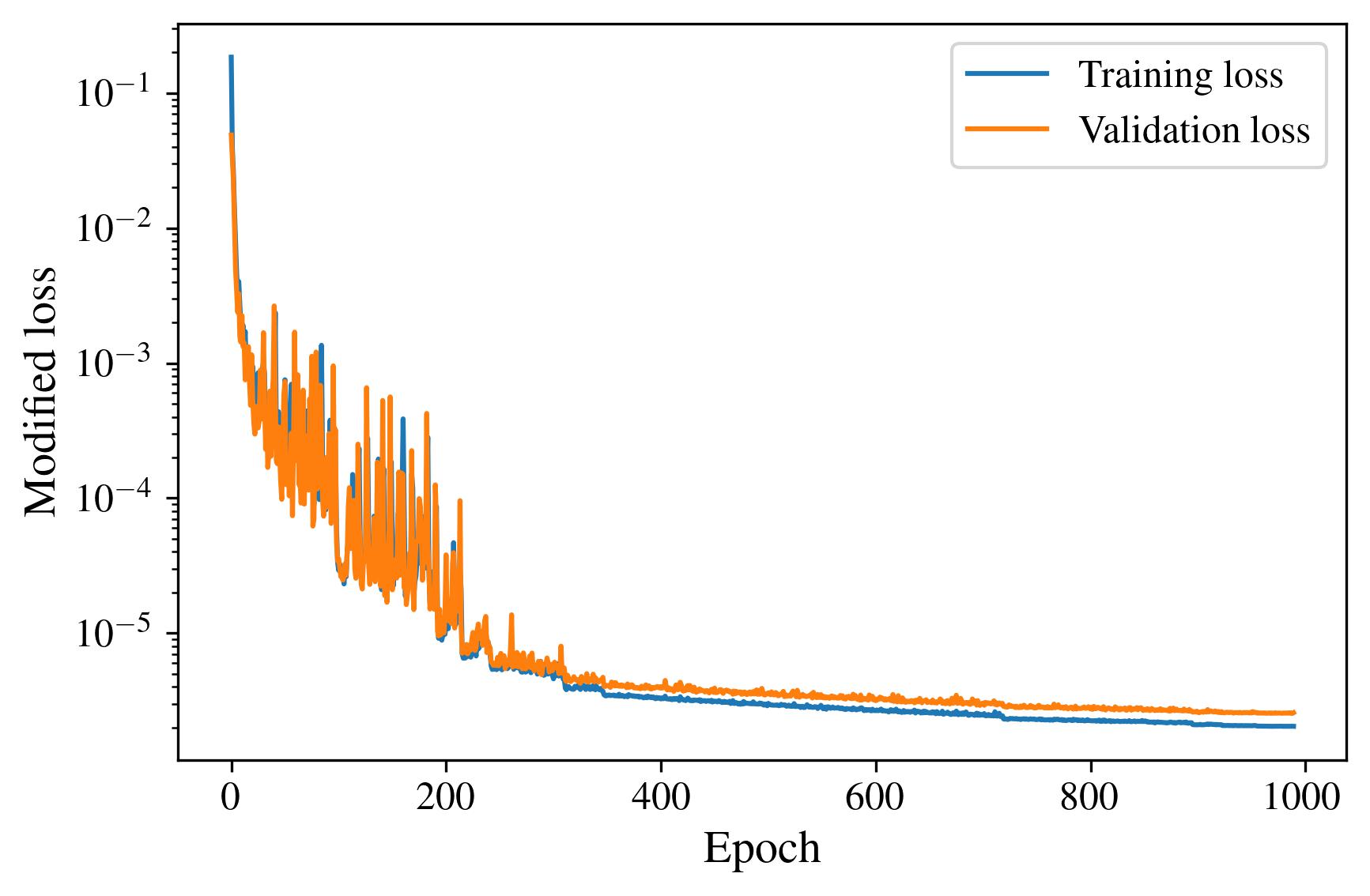}\\
    \caption{Training and validation loss histories under the modified loss function.}
    \label{fig:loss}
\end{figure} 

\subsection{Simulation results}
%
To evaluate the proposed neural output feedback operator, we consider the
test case $\nu =6.5$.  Without boundary control, the open-loop response diverges for the tested coefficient.
Fig.~\ref{fig:closed_loop_3d_comparison} compares the closed-loop responses obtained by the backstepping controller and the two neural  controllers. The responses $u_{\mathrm{bs}}$, $u_{\scriptscriptstyle \mathrm{MSE}}$, and $u_{\mathrm{mod}}$ are shown in Fig.~\ref{fig:closed_loop_3d_comparison} (a)-(c), and the corresponding
differences from the backstepping response are shown in
Fig.~\ref{fig:closed_loop_3d_comparison}(d) and (e). Both neural responses remain close to the backstepping response, while the modified loss yields smaller deviations.

\begin{figure*}[!htp]
    \centering
    \begin{tabular}{ccc}
        \includegraphics[width=0.32\textwidth,trim=0.1cm 0.9cm 0.7cm 2cm,clip]{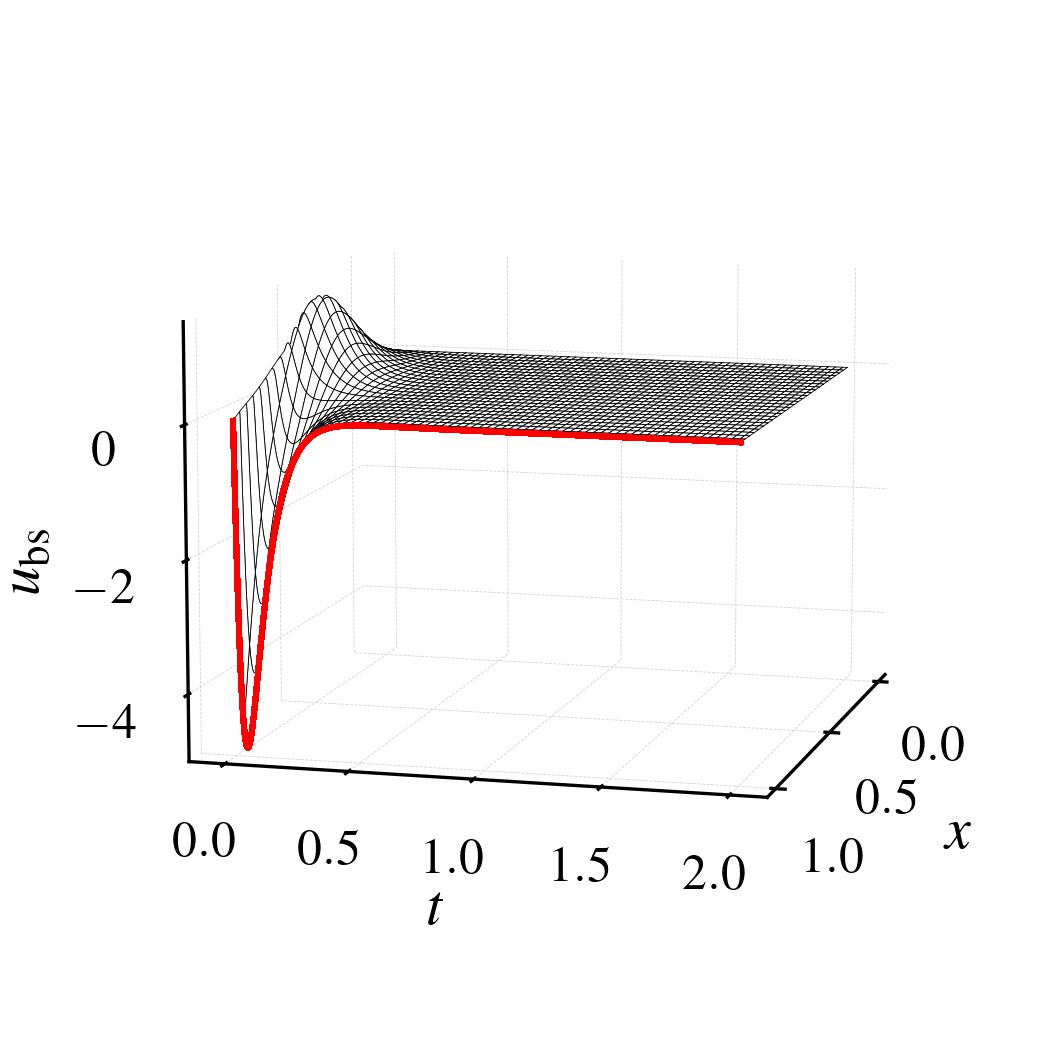}
        &
        \includegraphics[width=0.32\textwidth,trim=0.1cm 0.9cm 0.7cm 2cm,clip]{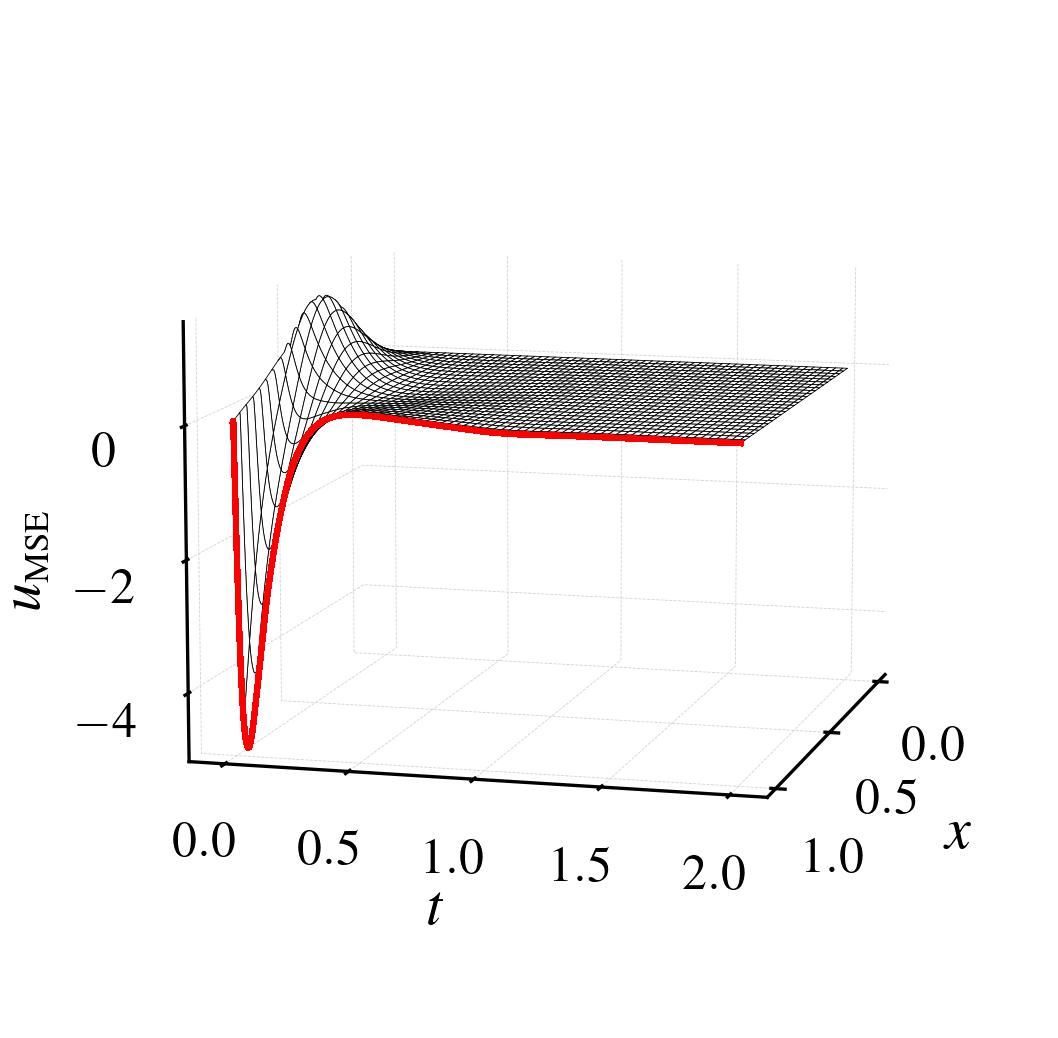}
        &
        \includegraphics[width=0.32\textwidth,trim=0.1cm 0.9cm 0.7cm 2cm,clip]{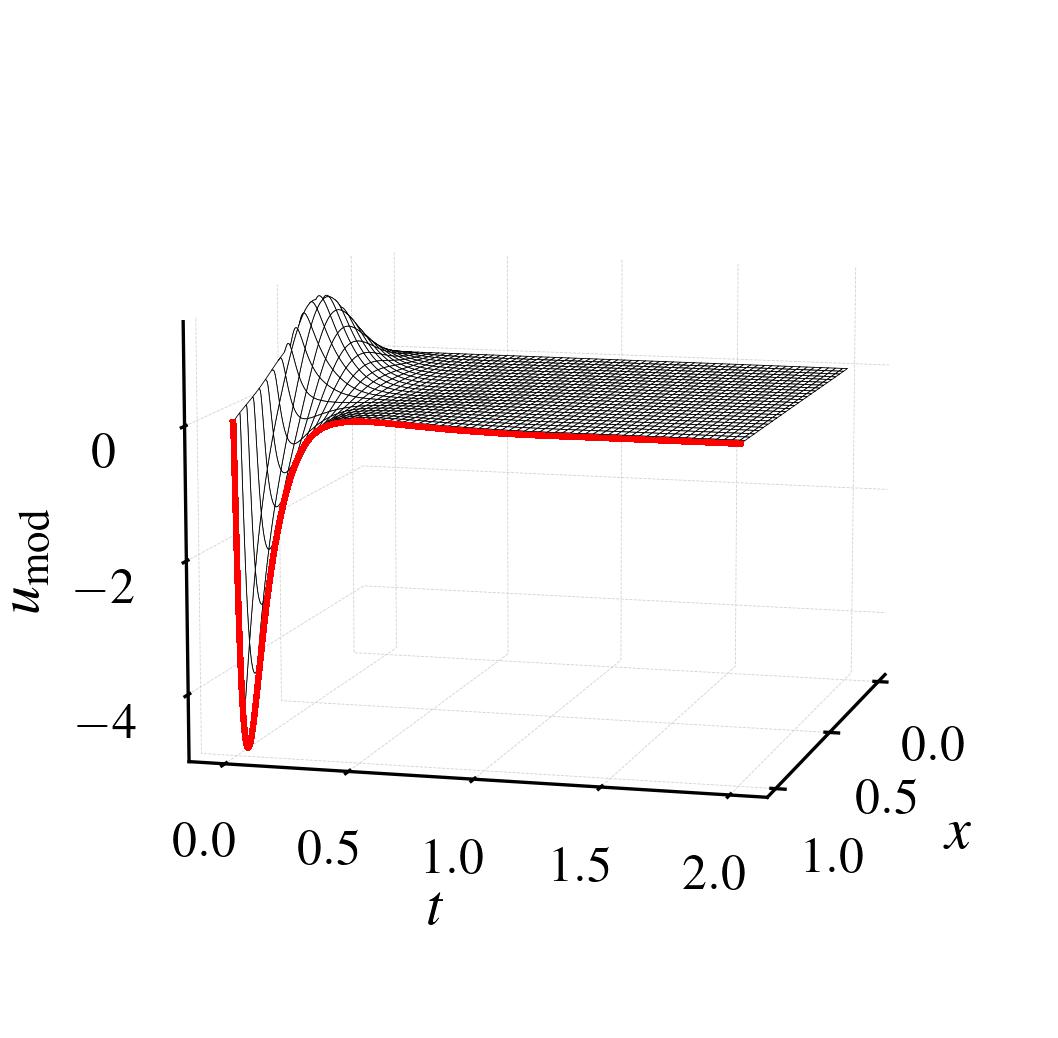}
        \\
        (a)  Response under backstepping $u_{\mathrm{bs}}$
        &
        (b) Response with MSE loss $u_{\scriptscriptstyle\mathrm{MSE}}$
        &
        (c) Response with modified loss $u_{\mathrm{mod}}$
    \end{tabular}
    \begin{tabular}{cc}
        \includegraphics[width=0.32\textwidth,trim=0.1cm 0.9cm 0.7cm 2cm,clip]{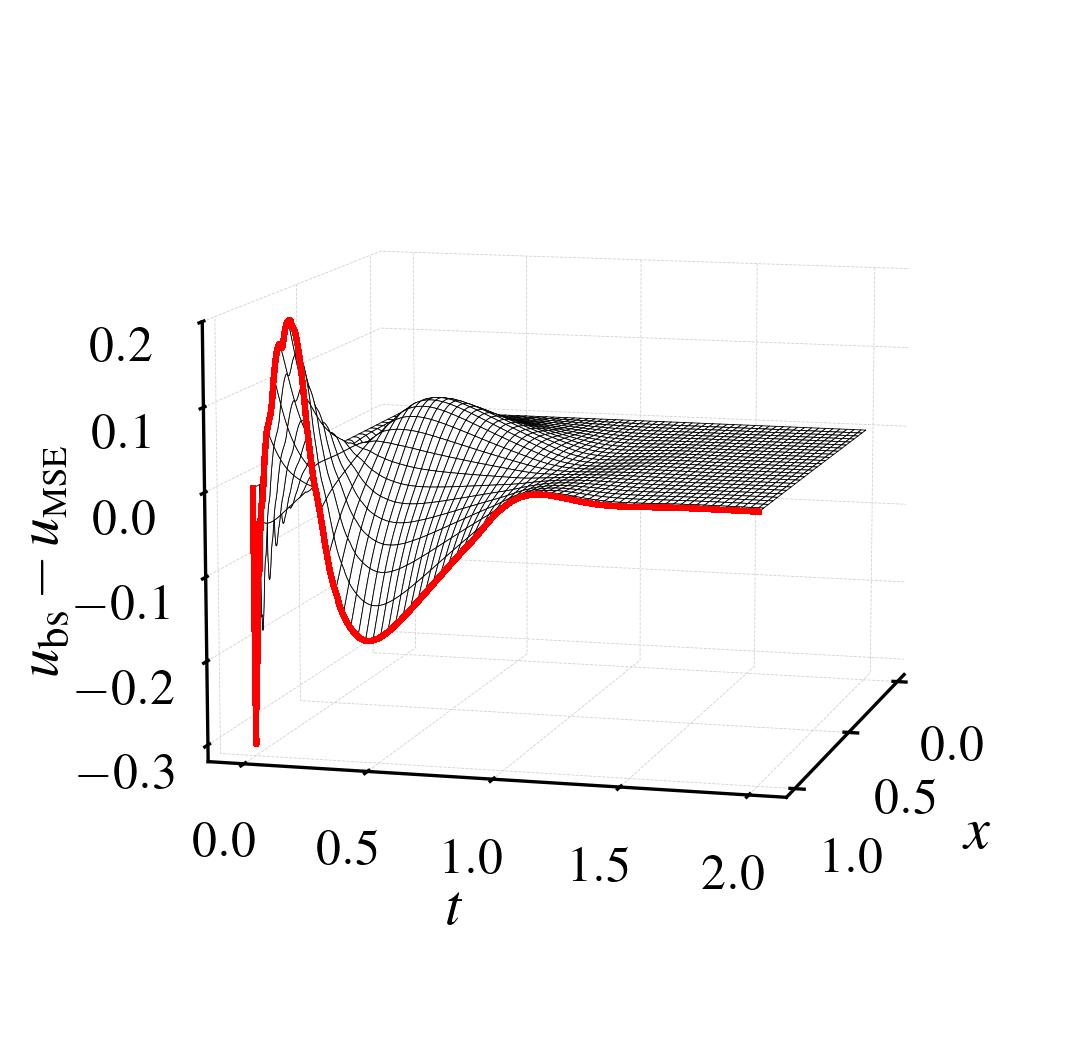}
        &
        \includegraphics[width=0.32\textwidth,trim=0.1cm 0.9cm 0.7cm 2cm,clip]{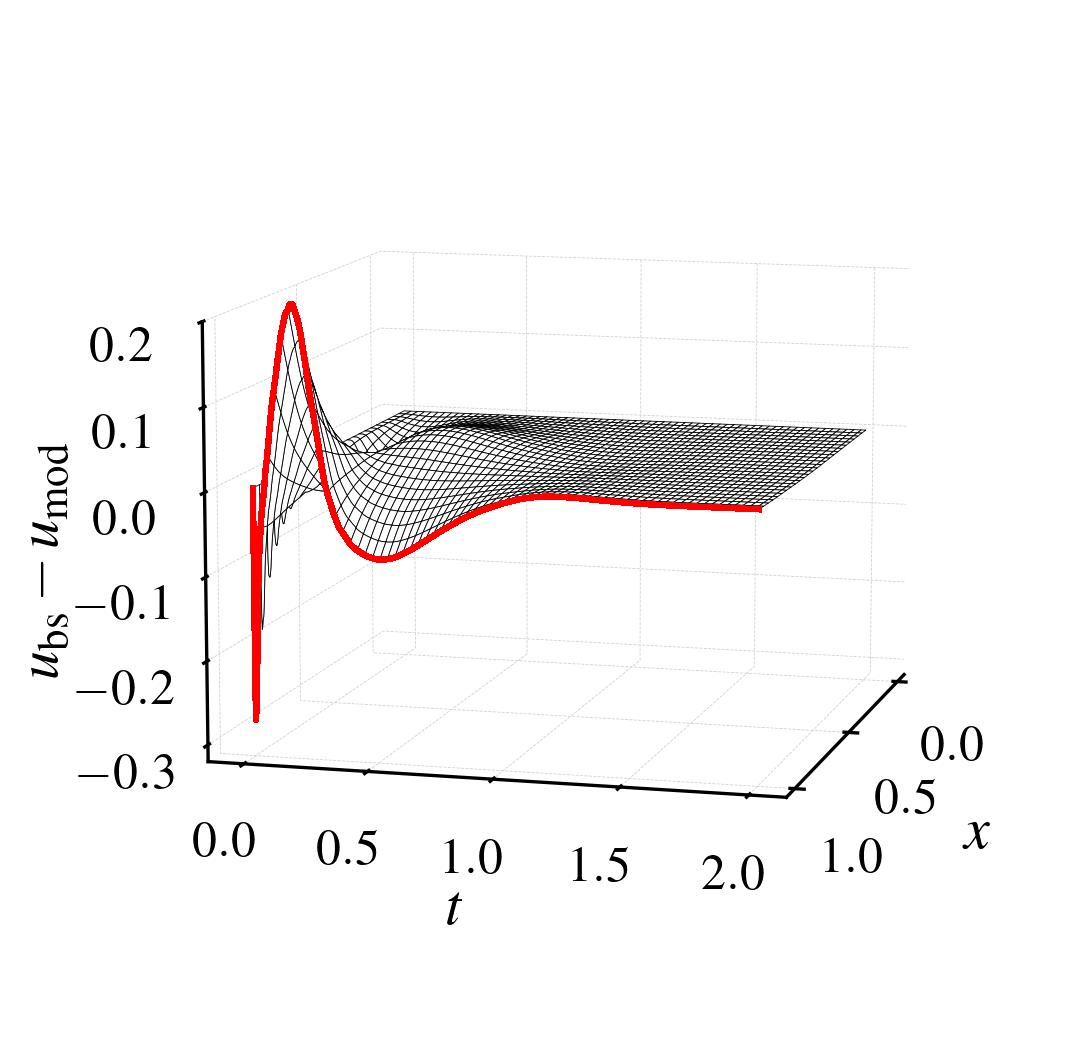}
        \\
        (d) Difference $u_{\mathrm{bs}}-u_{\scriptscriptstyle\mathrm{MSE}}$
        &
        (e) Difference $u_{\mathrm{bs}}-u_{\mathrm{mod}}$
    \end{tabular}
    \caption{Comparison of the closed-loop responses. (a)--(c) show the responses under the backstepping controller and the neural controllers trained with MSE and modified losses. (d) and (e) show the differences between the backstepping response and the two neural closed-loop responses.}
    \label{fig:closed_loop_3d_comparison}
\end{figure*} 
Fig.~\ref{fig:control_L2_comparison} compares the boundary control input and the $L^2$ norm of the closed-loop response. 
Both neural controllers follow
the backstepping input and stabilize the system. Compared with the MSE-loss
model, the modified-loss model produces a smoother control input, as shown
in the magnified inset, and a smoother $L^2$-norm trajectory.
\begin{figure*}[!htp]
    \centering
    \begin{tabular}{cc}
        \includegraphics[width=0.40\linewidth,trim=0.1cm 0.3cm 0.1cm 0.1cm,clip]{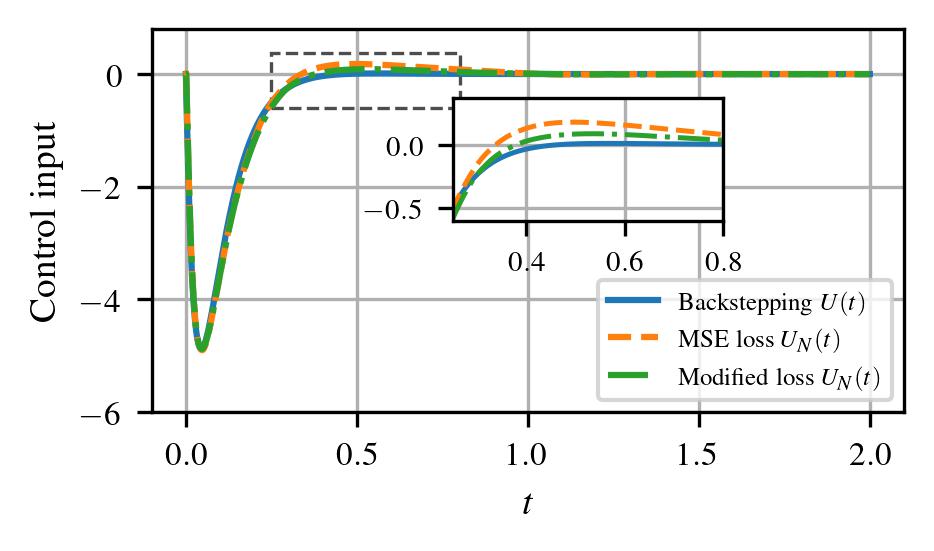}
        &
        \includegraphics[width=0.40\linewidth,trim=0.1cm 0.3cm 0.1cm 0.1cm,clip]{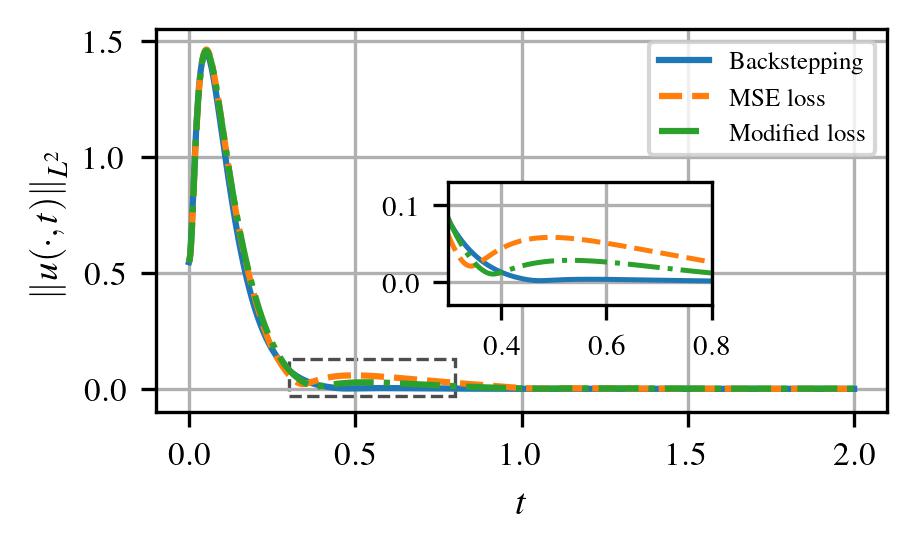}
        \\
        (a) Boundary control input
        &
        (b) Evolution of the $L^2$ norm of the closed-loop system.
    \end{tabular}
    \caption{Control input and closed-loop $L^2$ norm under the backstepping controller and neural controllers trained with MSE and modified losses.}
    \label{fig:control_L2_comparison}
\end{figure*}

Table~\ref{tab:performance} compares the closed-loop performance for three values of  $\nu$. 
For each neural controller $\star\in\{\mathrm{MSE},\mathrm{mod}\}$, let $u_\star$ and $U_\star$ denote the system state and control input.   The relative state and control errors are defined as 
\begin{equation} E_u^\star = \frac{\|u_\star-u_{\rm bs}\|_{L^2(Q_T)}} {\|u_{\rm bs}\|_{L^2(Q_T)}}, \qquad E_U^\star = \frac{\|U_\star-U_{\rm bs}\|_{L^2(I_T)}} {\|U_{\rm bs}\|_{L^2(I_T)}} , \nonumber \end{equation}
where $Q_T=(0,1)\times(0,T)$ and $I_T=(0,T)$.

\begin{table}[htbp]
\centering
\caption{Closed-loop performance and errors under different values of $\nu$}\label{tab:performance}
\begin{tabular}{ccccc}
\toprule
$\nu$ & controller 
& $\|u(\cdot,T)\|_{L^2}$ 
&  \makecell{relative\\state error $E_u$} 
& \makecell{relative\\control error $E_U$}\\
\midrule

\multirow{3}{*}{4.3}
  & Backstepping & 7.828e-07 & -- & -- \\
  & MSE loss     & 4.447e-02  & 2.533e-02 & 2.535e-02 \\
  & modified loss & \textbf{1.521e-03}   & \textbf{2.027e-02} & \textbf{2.040e-02} \\

\midrule

\multirow{3}{*}{6.5}
  & Backstepping & 3.779e-08 & -- & -- \\
  & MSE loss  & \textbf{1.014e-04} & 8.736e-02 & 8.430e-02   \\
  & modified loss &  8.233e-04 & \textbf{6.181e-02} & \textbf{6.107e-02} \\

\midrule

\multirow{3}{*}{8.8}
  & Backstepping & 1.041e-08 & -- & --\\
  & MSE loss    & 3.823e-04 & 4.083e-02 & 1.023e-01 \\
  & modified loss & \textbf{1.495e-04}   & \textbf{1.018e-02} &\textbf{4.738e-02}   \\

\bottomrule
\end{tabular}
\end{table}
The output-feedback backstepping controller is used as the baseline.  The bold values indicate the smaller value between the two neural controllers.  Although the modified loss does not yield the best value for every metric, it generally improves the approximation of the backstepping closed-loop dynamics.

\begin{rmk}
As an ablation study, we also trained a standard DeepONet to approximate an instantaneous map from the current measurement to the boundary control input, namely
$(\lambda,p(t))\mapsto U(t)$, 
without using the past measurement history. The resulting controller did not reliably stabilize the closed-loop system in our simulations. This observation is consistent with the causal nature of the observer-based output feedback law: the control input is determined by the observer state, which depends on the accumulated measurement history rather than only on the instantaneous value of $p(t)$.
\end{rmk}

\section{Conclusion}\label{sec:conc}
This paper proposes a DeepONet-LSTM neural operator for output feedback boundary control of reaction diffusion PDEs. The method learns a causal boundary operator from the reaction coefficient and the boundary measurement history to the control input, avoiding the online computation of the backstepping controller and observer kernels. We prove the Lipschitz continuity of the output feedback operator and establish practical stability of the closed-loop system under bounded neural approximation error. Simulations show that the learned controller stabilizes the system and produces responses close to those of the backstepping controller. Future work could extend the proposed framework to other classes of PDE systems, such as hyperbolic PDEs and high-dimensional PDEs.


\printbibliography

\end{document}